\documentclass{article}
\usepackage{booktabs} 
\usepackage[ruled]{algorithm2e} 
\usepackage{amsmath,amsthm,amsfonts}
\usepackage{amssymb}
\usepackage{authblk}

\usepackage{fullpage}

\usepackage{tikz}
\usetikzlibrary{calc}

\SetAlFnt{\small}
\SetAlCapFnt{\small}
\SetAlCapNameFnt{\small}
\SetAlCapHSkip{0pt}
\IncMargin{-\parindent}

\usepackage{natbib}
\setcitestyle{authoryear}

\usepackage{hyperref}

\newtheorem{lemma}{Lemma}
\newtheorem{proposition}{Proposition}
\newtheorem{theorem}{Theorem}
\newtheorem{observation}{Observation}
\newtheorem*{theorem*}{Theorem}
\newtheorem*{proposition*}{Proposition}
\newtheorem{openprob}{Open Problem}
\newtheorem*{openprob*}{Open Problem}

\theoremstyle{definition}
\newtheorem{definition}{Definition}

\theoremstyle{remark}
\newtheorem{example}{Example}

\usepackage{xcolor}
\newif\ifshowcomments
\showcommentsfalse
\newcommand{\docomment}[3]{\ifshowcomments \textcolor{#1}{[ #2 : #3 ]} \fi}
\newcommand{\bo}[1]{\docomment{brown}{Bo}{#1}}
\newcommand{\robin}[1]{\docomment{blue}{Robin}{#1}}

\DeclareMathOperator*{\E}{\mathbb{E}}
\newcommand{\reals}{\mathbb{R}}

\newcommand{\Welf}{\mathrm{Welfare}}
\newcommand{\Opt}{\textsc{Opt}}
\newcommand{\D}{\mathcal{D}}

\begin{document}

\title{High-Welfare Matching Markets via Descending Price}
\author[1]{Robin Bowers}
\author[1]{Bo Waggoner}
\affil[1]{University of Colorado, Boulder}

\maketitle

\begin{abstract}
  We consider design of monetary mechanisms for two-sided matching.
  Mechanisms in the tradition of the deferred acceptance algorithm, even in variants incorporating money, tend to focus on the criterion of stability.
  Instead, in this work we seek a simple auction-inspired mechanism with social welfare guarantees.
  We consider a descending-price mechanism called the Marshallian Match, proposed (but not analyzed) by \citet{waggoner2019matching}.
  When all values for potential matches are positive, we show the Marshallian Match with a ``rebate'' payment rule achieves constant price of anarchy.
  This result extends to models with costs for acquiring information about one's values, and also to matching on hypergraphs.
  With possibly-negative valuations, which capture e.g. job markets, the problem becomes harder.
  We introduce notions of approximate stability and show that they have beneficial welfare implications.
  However, the main problem of proving constant factor welfare guarantees in ``ex ante stable equilibrium'' remains open.
\end{abstract}

%
%
%
%
%

\section{Introduction}
A primary goal of designing mechanisms is to coordinate groups to arrive at collectively good allocations or outcomes.
For example, in auctioning a set of items to unit-demand buyers, the problem is to coordinate among the varied preferences of the buyers to achieve an overall good matching of buyers to items.
In such auction settings, ``good'' is usually formalized via \emph{price of anarchy}~\citep{roughgarden2017price}: in any equilibrium of the auction game, the expected social welfare (total utility) should be approximately optimal.

Matching people to people, with preferences on both sides, appears to require even more coordination.
\citet{gale1962college} introduced the foundational deferred-acceptance algorithm -- a matching mechanism without money -- for participants with ordinal preferences.
The key ``good'' property it achieves (if all participants are truthful) is \emph{stability}: no pair prefers to switch away from the given matching and match to each other instead.
Variants of deferred acceptance have had significant impact in applications from kidney exchange to the National Residency Matching Program (NRMP) for doctors and hospitals~(e.g. \citet{iwama2008survey}).

We are motivated by two drawbacks of deferred-acceptance-style approaches.
First, the \emph{social welfare} generated by such mechanisms is unclear, even in settings where money is explicitly modeled such as matching with contracts~\citep{hatfield2005matching}.
While their criterion of stability is a nice property, its relationship to welfare is not obvious.
We would like to investigate this relationship and obtain explicit welfare guarantees.

Second, it is unclear the extent to which such mechanisms are compatible with \emph{inspection} stages in which participants must invest effort to discover their preferences.
For example, in practice, the design of the NRMP requires relatively expensive and time-constrained interviews, which must be completed before matching begins.
While such concerns have motivated significant work on information acquisition in matching markets, particularly variants of deferred acceptance~(e.g. \cite{immorlica2021designing}; see Section \ref{sec:related-work}), to our knowledge none of it incorporates monetary mechanisms with quantitative welfare guarantees.
On the other hand, prior work of \cite{kleinberg2016descending} has shown that even in the special case of matching people to items (which have no preferences), approximately optimal welfare \emph{requires} a market design with dynamically interspersed matching and information acquisition.
\cite{kleinberg2016descending} showed that descending-price mechanisms tend to be compatible with costly inspection stages and still yield high social welfare, due to a connection with the Pandora's box problem~\citep{weitzman1979optimal}.

\paragraph{The $1/4$-rebate Marshallian Match.}
Inspired by \citet{kleinberg2016descending}, \citet{waggoner2019matching} propose the ``Marshallian Match'' (MM) for two-sided matching with money.
In \citet{marshall1920principles}, its namesake describes a theory of market clearing in which the buyer-seller matches that generate the largest surplus -- i.e. the most net utility between the pair -- occur first, and so on down.\footnote{This dynamic eventually leads to the market clearing price (e.g. \citet{plott2013marshall}), after which no more positive-surplus matches are possible. Indeed, in a simple commodity market it is possible to ignore Marshall's dynamics and focus on the calculation of the clearing price. But in a more complex two-sided matching problem, we appear to require dynamics in order to properly coordinate matches.}
A similar dynamic is observed in decentralized matching markets~\cite{chade2017sorting}, yet it does not directly underly standard centralized designs such as deferred acceptance.

Similarly, the MM begins with a high price that descends over time.
Participants maintain a bid on each of their potential matches, with the sum of the two bids ideally representing the total surplus generated by the match.
When the price reaches the sum of any pair's bids on each other, that pair is matched.
They pay their bids and drop out of the mechanism, which continues.
In the ``$1/4$-rebate'' variant studied in this paper, the mechanism only keeps half of the sum of the bids and the participants split the other half, each receiving a rebate of $1/4$ of the sum of the bids.
\citet{waggoner2019matching} speculate on the dynamics, strategy, and benefits of this mechanism and variants,\robin{in the absence of rebates\bo{added ``and variants''. I recall mentioning rebates somewhere in there.}} but do not obtain theoretical results.

\subsection{Our results}

\paragraph{Nonnegative values.}
We first consider a setting where all participants' values are nonnegative.
Under this restriction, we show a general \emph{price of anarchy (PoA)} guarantee for the Marshallian Match, i.e. in any Bayes-Nash equilibrium the expected social welfare is within a constant factor of the optimal possible.
The positive-bids setting can model, for example, matching of industrial plants to geographic areas (an application of \cite{koopmans1957assignment}) or matching local businesses to municipal-owned locations.

Next, we extend the result to a \emph{group formation} setting, which also models \emph{matching on hypergraphs}.
Agents must be partitioned into subsets of size at most $k$, with private valuations for joining each possible subset.
We modify the MM by clearing a subset when the price reaches the sum of its bids.
We obtain a $\Omega(1/k^2)$ price of anarchy for this problem.

Next, we show that the welfare guarantee also extends to a model with \emph{information acquisition costs}.
Often, participants do not initially know their values for a potential match.
It requires time, effort, and/or money to investigate and discover one's value.
A good mechanism should carefully coordinate these investigations to happen at appropriate times, or else significant welfare will be lost: participants will either waste too much utility on unnecessary inspections, or they will forego valuable matches due to the inspection cost and uncertainty about the match.
Here, although the optimal first-best is unknown and likely NP-hard, we obtain the same PoA guarantee.

\begin{theorem*}
  With nonnegative values, the $1/4$-rebate Marshallian Match has the following guarantees:
  \begin{enumerate}
    \item For matching on general graphs, a Bayes-Nash price of anarchy of at least $1/8$.
    \item For matching on hypergraphs with group size at most $k$, a Bayes-Nash price of anarchy of at least $\tfrac{1}{2k^2}$.
    \item For matching on general graphs with inspection costs, a Bayes-Nash price of anarchy of at least $1/8$.
  \end{enumerate}
\end{theorem*}

\paragraph{Proof ingredients.}
The proofs rely on a smoothness approach (e.g. \cite{roughgarden2017price}) along with several key properties of our variant of the MM.
First, the MM limits information leakage: a participant cannot observe others' strategies until they themselves are matched, at which point it is too late to react.
This controls the otherwise complex strategic interactions of dynamic mechanisms.

Next, the ``rebate'' payment rule of the MM crucially allows participants to align their personal utility with the order of market clearing.
In particular, a participant who deviates to truthful bidding will always receive utility equal to their rebate, which equals $1/4$ of the current descending price.
So the earlier the participant is cleared, the higher their utility.
This is also beneficial to social welfare, where higher-surplus edges should generally be matched first.

Finally, the descending-price structure of the MM is compatible with information acquisition and the ``Pandora's Box'' problem~\citep{weitzman1979optimal}.
Participants are able to manage the risk-reward tradeoff for investing in information acquisition, because once prices have descended to a low point, they know that they can lock in available matches for a bounded cost.
To prove PoA in this setting, we adapt techniques of \cite{kleinberg2016descending} for analyzing welfare in models of inspection.

\paragraph{Possibly-negative values.}
We then consider a general two-sided matching setting where values may be negative.
Understanding this setting is desirable because it more accurately captures job markets, where workers incur a cost (i.e. negative value) for being matched to a job and must be compensated more than that cost.
Here we see for natural reasons that a PoA result for the MM is impossible:
if all participants set their bids so as to refuse all matches, no single participant can deviate to cause any change and the equilibrium obtains zero welfare.
In general, this suggests that a kind of stability condition may be natural and necessary for high-welfare matching mechanisms of any kind.
We observe that \emph{approximate ex post} stability implies approximately optimal welfare, and show that the MM achieves approximate ex post stability if participants bid truthfully.

\begin{proposition*}
  In any $k$-approximate ex post stable strategy profile, the $1/4$-rebate Marshallian Match (in fact, \emph{any} mechanism) achieves at least $1/k$ of the optimal expected welfare.
  Truthful bidding in the $1/4$-rebate Marshallian Match is $4$-approximately ex post stable.
\end{proposition*}

However, truthfulness is not generally an equilibrium.\footnote{We note that even for deferred acceptance, which is also stable \emph{if all participants are truthful}, in general only one side of the market optimizes their outcomes by being truthful.}
This raises the question of whether it is reasonable to assume participants will adopt approximately stable strategy profiles.
We argue that ex post stability is too strong an assumption, and instead propose \emph{ex ante} stability.
We show that in a Nash setting (not Bayes-Nash) with fixed valuations, if strategies are approximately \emph{ex ante stable}, then the welfare of MM is approximately optimal.

\begin{theorem*}
  In the general Nash setting with fixed valuations, in any strategy profile that is $k$-approximate ex ante stable, the $1/4$-rebate Marshallian Match achieves at least a $\tfrac{1}{4k}$ fraction of the optimal welfare.
\end{theorem*}
In fact, the stability property also ensures that participants keep a large fraction of the welfare; the proof uses that participant surplus alone (i.e. welfare minus payments) is at least $\tfrac{1}{4k}$ of the optimal welfare.

Unfortunately, this stable-price-of-anarchy result is fragile and the proof does not extend to the \emph{Bayes-Nash} setting where private valuations are drawn from a common-knowledge prior.
This is roughly due to the difficulty of coordinating and communicating deviations between ``blocking pairs''.
Therefore, the main problem of proving a welfare guarantee in a general negative-bids setting and under a reasonable stability assumption remains open.

\begin{openprob*}
  Give a well-justified stability assumption and a mechanism such that, in the \emph{Bayes-Nash setting} with general values, stable strategy profiles guarantee a constant factor of the optimal expected welfare.
\end{openprob*}

\subsection{Related work} \label{sec:related-work}
We have not found any mechanisms in the literature involving a two-sided matching market with money and quantifiable welfare results.\footnote{One can always apply a general Vickrey-Clarke-Groves (VCG) mechanism, which has an equilibrium with optimal social welfare. But VCG is undesirable because it appears incompatible with both price of anarchy results and models with inspection costs, cf. \citet{kleinberg2016descending}.}
However, the literature on matching with strategic agents is very broad, including with inspection stages, and we highlight a number of papers that are related to our problem.

Probably closest to our work is \citet{immorlica2021designing}, which considers design of a platform to coordinate two-sided matchings with inspection costs and quantifiable welfare guarantees.
Motivated by e.g. matching platforms for romantic dating, the paper studies agents coming from specific populations with a known distribution of types.
The mechanism uses its knowledge of the type distributions to compute strategies for directing inspection stages (e.g. first dates).
The computational problem is challenging and intricate.
\citet{immorlica2021designing} is able to show the structure of equilibria and use this to give welfare guarantees, all in a setting without money.
In contrast, we are interested in monetary mechanisms, motivated (eventually) by e.g. labor markets.
We consider a very simple descending-price mechanism that has no access to knowledge about the agent types or distributions.

Beyond \citet{immorlica2021designing} there is an extensive literature on matching marketplaces and dynamics.
Work in that literature involving money (transferable utility) historically often takes a Walrasian equilibrium perspective, while ours is in the tradition of auction design.
We refer to the survey of \citet{chade2017sorting} for more on this literature.

As mentioned above, a number of recent works study matching markets with information acquisition, but generally consider variants of deferred acceptance without money.
Works of this kind include \citet{he2020application,che2019efficiency,ashlagi2020clearing,chen2021information,fernandez2021centralized,immorlica2020information,hakimov2021costly}.

Among these, \citet{immorlica2020information} uses a lens of optimal search theory similar to ours.
It studies matching of students to schools.
Its matching problem is almost one-sided, in the sense that schools have known and fixed preferences.
The focus is on coordinating efficient acquisition of information by students. 
Unlike our social-welfare perspective, that work focuses on the more standard criterion of stability and introduces regret-free stable outcomes.
In particular, it does not involve money.
Similarly, \citet{hakimov2021costly} study a serial-dictator mechanism for coordinating student inspection without money.
We refer to \citet{immorlica2020information} for an extensive discussion of further work related to information acquisition in matching markets.

Our notions of stability are naturally closely related to others in the literature.
Ex post stability is only a quantitative version of stability in matching; a more sophisticated version of it is used by \citet{immorlica2020information}, for example.
\cite{fernandez2021centralized} also utilizes a similar notion of stability, in a setting of incomplete information.


\section{Preliminaries} \label{sec:prelim}
We now define the model and the variant of the Marshallian Match mechanism studied in this paper, originally described by \citet{waggoner2019matching}.
We define a general model, in which the graph is possibly non-bipartite and bids and values are possibly negative.
Later sections will consider specific restrictions.

There is a finite set of $n$ agents, forming vertices of an undirected graph $G = (\{1,\dots,n\}, E)$.
For now, we do not assume that $G$ is bipartite.
The presence of an edge $\{i,j\} \in E$ represents that it is feasible to match agents $i$ and $j$.
In this case, agent $i$ has a \emph{value} $v_{ij} \in \reals$ for being matched to $j$, and symmetrically, $j$ has a value $v_{ji}$ for being matched to $i$.
If $i$ and $j$ are neighbors, we let $s_{ij} = v_{ij} + v_{ji}$ denote the \emph{surplus} of the edge $\{i,j\}$.

An agent $i$'s \emph{type} consists of their values $v_{ij}$ for each feasible partner $j$.
In the \emph{Nash} setting, each agent has a fixed type, and types are common knowledge.
In the better-motivated \emph{Bayes-Nash} setting, types are drawn from a common-knowledge joint distribution $\D$ and each agent observes their own type.

\subsection{The Marshallian Match}
In the MM, a global price $p(t)$ begins at $+\infty$, i.e. $p(0)=\infty$, and descends continuously in time until it reaches zero at time $1$, i.e. $p(1) = 0$.\footnote{This can be accomplished in theory by letting the price be e.g. $p(t) = \frac{1}{2t}$ for $t \in [0,1/2]$ and $p(t) = 2-2t$ for $t \in [1/2,1]$.}
Each agent $i$ maintains, for all neighbors $j$, a bid $b_{ij}(t)$ at each time $t$.
The mechanism can observe all bids at all times, but agents cannot observe any bids except their own.
For convenience, we may drop the dependence on $t$ from the bid notation.
When the sum of bids on any edge exceeds the global price, i.e. $b_{ij} + b_{ji} \geq p(t)$, then $i$ and $j$ are immediately matched to each other.
Each agent pays their respective bid to the mechanism.

The mechanism keeps half of this total payment and returns one-fourth to each player.
Therefore, we call this variant the $1/4$-rebate Marshallian Match.
We discuss other variants in Section \ref{sec:conclusion}.

\paragraph{Intuition for the mechanism.}
Why might this mechanism be good, and how might participants strategize?
We briefly describe some intuition, referring the reader to \citet{waggoner2019matching} for more detailed discussion.
Social welfare and price of anarchy will be formally defined below.

First, the edges of the graph are in competition with each other to match first.
When an edge $\{i,j\}$ is matched, it produces a surplus of $s_{ij} = v_{ij} + v_{ji}$.
The first-best solution, i.e. the optimal solution for a central planner who holds all information, is to select a maximum matching where $s_{ij}$ are the edge weights.
However, as discussed in \cite{kleinberg2016descending}, algorithms for maximum matching are complex and appear to interact poorly with inspection stages.
More robust is an approximate first-best solution: the \emph{greedy} matching, where the highest-weight edge is matched first, and so on down.
This procedure obtains at least half of the optimal welfare, and can be simulated by a Marshallian Match in which all participants bid truthfully.
In \cite{kleinberg2016descending}, this fact was used to obtain a constant price of anarchy for matching people to items, including in the presence of inspection costs.

However, two-sided matching introduces new strategic considerations.
``Within'' an edge $\{i,j\}$, there is a competition or bargaining for how to split the surplus generated.
An agent $i$ with many edges (many outside options) may be able to underbid significantly, while her counterpart $j$ with very few outside options must offer a very high bid along that edge.
The question is whether this strategizing and competition between $i$ and $j$ destroys the cooperation incentives for the overall bid on the edge $\{i,j\}$.

\subsection{Notation, strategies, and welfare}
A strategy $b_i$ for player $i$ consists of a plan\footnote{We observe that usual nuances around equilibrium in dynamic games, such as non-credible threats, refinements such as subgame perfect equilibrium, etc., do not arise here. In our variant of the MM, each agent $i$ observes nothing until they are matched, after which they can no longer affect the game. So without loss of generality, $i$ commits in advance to a plan $\{ b_{ij}(t) \}$ and follows it until matched.} $b_{ij}(t)$ for how to set bids over time, for each neighbor $j$.

A full strategy profile is denoted $b = (b_1,\dots,b_n)$.
We let $v_i(b)$ be $i$'s value for their match when $b$ is played (zero if none), a random variable depending on the randomization in the strategies and, in the Bayes-Nash setting, on the random draw of the types.
Next, $\pi_i(b)$ denotes $i$'s net payment to the mechanism, i.e. bid minus rebate.
Finally, $p_i(b)$ denotes the \emph{total} net payment on the edge that $i$ is matched along, i.e. $p_i(b) = \pi_i(b) + \pi_j(b)$ when $i$ is matched to $j$.

We assume all agents are Bayesian, rational, and have preferences quasilinear in payment.
That is, given a mechanism and a particular strategy profile $b$, the \emph{utility} of a participant $i$ is the random variable
  \[ u_i(b) = v_i(b) - \pi_i(b) . \]
For example, if under profile $b$, $i$ matches to $j$ at time $t$, then $\pi_i(b) = b_{ij}(t) - \frac{p(t)}{4}$ and $u_i(b) = v_{ij} - b_{ij}(t) + \frac{p(t)}{4}$.

A \emph{Nash equilibrium} of a mechanism with given types $\{v_{ij}\}$ is a strategy profile $b$, consisting of a plan for how to set bids at each moment in time, where each participant maximizes their expected utility, i.e.
  \[ \E u_i(b) \geq \E u_i(b_{-i}, b_i') \]
for all $i$ and for any other strategy $b_i'$ of $i$, where the randomness is taken over the strategies.
A \emph{Bayes-Nash equilibrium} is defined in exactly the same way, but in a setting consisting of a joint distribution over types.
In that case, the randomness is taken over both types and strategies (which are maps from an agent's type to a plan for bidding over time).

We use $\Welf(b)$ to denote the expected social welfare, i.e. sum of utilities and payments:
  \[ \Welf(b) = \E \left[ \sum_i v_i(b) + \sum_j v_j(b) \right] , \]
where the expectation is over all randomness.
In the setting with inspection costs, utility and social welfare also accounts for the loss of utility from the inspection processes; this will be formalized at the relevant point, Section \ref{sec:inspection}.

$\Welf(\Opt)$ refers to the optimal social welfare.
In the Nash setting,
	\[ \Welf(\Opt) = \max_M \sum_{\{i,j\} \in M} v_{ij} + v_{ji} , \]
where the maximum is over all matchings $M$.
In the Bayes-Nash setting,
	\[ \Welf(\Opt) = \E \left[ \max_M \sum_{\{i,j\} \in M} v_{ij} + v_{ji} \right], \]
where the expectation is over the realizations of types.

The \emph{price of anarchy} measures the worst-case ratio of $\Welf(b)$ to $\Welf(\Opt)$ in any equilibrium $b$.
For Nash equilibrium, we have
  \[ \text{PoA} = \min \frac{\Welf(b)}{\Welf(\Opt)} ,\]
where the minimum is taken over all settings (i.e. all types of the participants) and all Nash equilibria $b$.
The Bayes-Nash price of anarchy is defined in exactly the same way, but the minimum is now over all Bayes-Nash settings (i.e. joint distributions on types) and all Bayes-Nash equilibrium strategy profiles $b$.
Note that a Nash equilibrium is a special case of Bayes-Nash where the type distributions are degenerate.
Therefore, a Bayes-Nash price of anarchy result immediately implies a Nash price of anarchy.



%

\section{Results for Positive Valuations}
In this section, we consider a restriction of the general setting where all values $v_{ij}$ are nonnegative.
First, we show that the $1/4$-rebate Marshallian Match achieves a constant approximation of optimal welfare for matching.
Second, we extend the result to the group formation (i.e. matchings on hypergraphs) setting.
Finally, we extend the result in a different direction to the case where participants do not initially know their valuations and can choose to expend effort to discover them.

\subsection{The vanilla positive valuations model}
Here, we take the general model of Section \ref{sec:prelim} and assume that each valuation satisfies $v_{ij} \geq 0$.
We modify the Marshallian Match to require all bids $b_{ij}$ to be nonnegative at all times.

\paragraph{Intuition.}
The nonnegative MM is similar to running multiple interlocking descending-price unit-demand auctions simultaneously.\footnote{Although general price of anarchy results are available for these kinds of auctions for goods, e.g. \citet{lucier2015greedy,feldman2016price}, we do not know of any that apply to two-sided matching.}
This parallel is most obvious in the bipartite case: any bidder $i$ competes against other bidders in the same set for her favorite matches.
Extending this perspective, each bidder $i$ could hypothetically bid as though her neighbors $j$ were simply items with no preferences.
By analyzing the failure of this hypothetical strategy, we find in the MM that it is connected with a different high-welfare event, namely $j$ matching early.
This intuition underlies our smoothness lemma, discussed next.

\paragraph{Smoothness for two-sided matching.}
We give a smoothness lemma that powers our price of anarchy result.
Recall (e.g. \citet{roughgarden2017price}) that smoothness proofs of PoA proceed by guaranteeing high-welfare events in a counterfactual world where $i$ deviates to a less-preferred strategy.
One challenge is that in a dynamic mechanism that proceeds over time, a deviation could cause chain reactions that make it impossible to reason about the outcomes.
Here, we rely on our variant of MM that does not leak any information about bids or strategies.
The only piece of information the agent receives from the mechanism comes at the moment they are matched, after which they cannot react.

A second key challenge is that in a matching market, $i$ may not match $j$ \emph{and} the prices that $i$ and $j$ each pay may still be low, obstructing an adaptation of a standard smoothness proof.
In the MM there is, however, a high \emph{total} price paid for an edge that obstructs the match.
Recall that while $\pi_i(b)$ is $i$'s payment, $p_i(b)$ is the \emph{total} payment on the edge containing $i$ that is matched (zero if $i$ is unmatched).

\paragraph{The deviation and smoothness lemma.}
Let $b$ be any strategy profile. For any bidder $i$, define the deviation strategy $b_i'$ as $b_{ij}'(t) = v_{ij}$ for all feasible neighbors $j$ and all times $t$.
That is, the deviation strategy is simply truthful bidding.

Recall that for the $1/4$-rebate MM, $i$'s utility in deviation when matching to $j$ at bid $b_{ij}'(t)$ and price $p(t)$ is
\begin{align*}
  u_i(b_i',b_{-i}) &= v_{ij} - b_{ij}'(t) + \frac{p(t)}{4}  \\
         &= \frac{p(t)}{4} .
\end{align*}

\begin{lemma} \label{lemma:nonneg-smooth}
	In the nonnegative values setting, for any feasible pair $\{i,j\}$, any strategy profile $b$, and any realization of types,
	\begin{align}
		u_i(b_i',b_{-i}) + \frac{p_i(b)}{4} + \frac{p_j(b)}{4} &\geq \frac{v_{ij}}{8} .
	\end{align}
\end{lemma}
\begin{proof}
	All three terms on the left-hand side are nonnegative.
	Therefore, if $p_j(b) \geq v_{ij}/2$ or $p_i(b)\geq v_{ij}/2$, the result is immediate.

	Otherwise, under $b$, neither $i$ nor $j$ is matched on an edge with net payment at least $v_{ij}/2$.
	So they are both unmatched by the time the price has dropped to $p(t) = v_{ij}$.
	The behavior of the mechanism and all participants is identical under $b$ and $(b_i',b_{-i})$ until $i$ matches, because the two profiles cannot be distinguished.
	So if $i$ is still unmatched under $(b_i',b_{-i})$ when $p(t) = v_{ij} = b_{ij}'(t)$, then $i$ matches to $j$ at this price.
	We conclude that $i$ matches at some price $p(t) \geq v_{ij}$.
	Therefore, $u_i(b_{-i},b_i') \geq p(t)/4 \geq v_{ij}/4$.
\end{proof}

\begin{theorem} \label{thm:nonneg-welfare}
	In the nonnegative values setting, the $1/4$-rebate Marshallian Match has a Bayes-Nash price of anarchy of at least $1/8$.
\end{theorem}
\begin{proof}
	Letting $b$ be any Bayes-Nash equilibrium and $M^*$ be the optimal matching (a random variable), we have the following.
	We will use that, in equilibrium, $i$ prefers $b_i$ to $b_i'$; that $M^*$ contains at most every participant, and utilities and payments are nonnegative; and Lemma \ref{lemma:nonneg-smooth}.
	\begin{align*}
		\Welf(b) &= \E \sum_i \left(u_i(b) + \frac{p_i(b)}{2}\right)  \\
		&\geq \E \sum_i \left(u_i(b_{-i},b_i') + \frac{p_i(b)}{2}\right) \\
		&\geq \E\left[\sum_{\{i,j\}\in M^*} \left(u_i(b_{-i},b_i') + u_j(b_{-j},b_j') + \frac{p_{i}(b)}{2} + \frac{p_{j}(b)}{2}\right)\right]\\
		&= \E\left[\sum_{\{i,j\}\in M^*} \left(u_i(b_{-i},b_i') + \frac{p_{i}(b)}{4} + \frac{p_{j}(b)}{4} + u_j(b_{-j},b_j') + \frac{p_{i}(b)}{4} + \frac{p_{j}(b)}{4} \right) \right]\\
		&\geq \E\left[\sum_{\{i,j\}\in M^*} \left( \frac{v_{ij}}{8} + \frac{v_{ji}}{8} \right) \right] = \frac{1}{8}\Welf(\Opt) .
	\end{align*}
\end{proof}

\subsection{Matchings on hypergraphs}
We now extend to the problem of coordinating formation of groups of size up to $k$.
We will be brief because the proof is similar to the matching case above, i.e. the special case where $k=2$ and groups of size one are disallowed.

Instead of a graph, we are given a hypergraph where agents are vertices and a hyperedge $S$ represents a feasible group, i.e. subset of agents of size at most $k$.
The value of agent $i$ for being assigned to group $S$ is $v_{iS} \geq 0$.
A ``matching'' or assignment $M$ consists of a subset of the hyperedges such that no agent is in two different groups $S,S'  \in M$.
The surplus of a group $S$ is $s_S = \sum_{i \in S} v_{iS}$.
The social welfare of an assignment $M$ is $\sum_{S \in M} s_S$.

The $1/4$-rebate Marshallian Match for this setting is modified as follows to a ``$\tfrac{1}{2k}$-rebate MM''.
Participants $i$ maintain bids $b_{iS}$ on all of their feasible groups $S$.
When the descending price matches the total sum of bids on any group, i.e. $p(t) = \sum_{i \in S} b_{iS}$, that group is matched and drops out of the mechanism.
All members pay their bids, and each receives a rebate of $\tfrac{p(t)}{2|S|}$, with the mechanism keeping $\frac{p(t)}{2}$.

Again consider the truthful deviation $b_i'$ where $b_{iS}'(t) = v_{iS}$ for all feasible $S$ and all $t$.
When $i$ matches in group $S$ at price $p(t)$ in deviation,
\begin{align*}
	u_i(b_i',b_{-i}) &= v_{iS} - b_{iS}'(t) + \frac{p(t)}{2|S|}  \\
	&\geq \frac{p(t)}{2k} .
\end{align*}
\begin{lemma} \label{lemma:group-smoothness}
	For any strategy profile $b$, realizations of types, agent $i$, and hyperedge $S$ containing $i$,
	\[ u_i(b_i',b_{-i}) + \frac{1}{k^2}\sum_{j \in S} p_j(b) \geq \frac{v_{iS}}{2k^2} . \]
\end{lemma}
\begin{proof}
	All terms are nonnegative.
	If under $b$ there exists $j \in S$ with $p_j(b) \geq \frac{v_{iS}}{2}$, then we are done.
	Otherwise, in $(b_{-i},b_i')$ all members of $S$ are unmatched at least until $i$ matches, which at the latest occurs at $p(t) = v_{iS}$, yielding $i$ utility at least equal to its rebate of $\tfrac{p(t)}{2|S|} \geq \tfrac{v_{iS}}{2k} \geq \tfrac{v_{iS}}{2k^2}$.
\end{proof}

\begin{theorem} \label{thm:group-welfare}
	In the hyperedge matching (group formation) setting with nonnegative values and group size up to $k$, the Marshallian Match has a Bayes-Nash price of anarchy of at least $\frac{1}{2k^2}$.
\end{theorem}
\begin{proof}
	Letting $b$ be a Bayes-Nash equilibrium and $M^*$ the optimal assignment (a random variable), we have the following.
	The first line follows because groups have size at most $k$, so the total payment of a group $S$ is at least $\sum_{i \in S} p_i(b)/k$.
	\begin{align*}
		\Welf(b) &\geq \E \sum_i \left( u_i(b) + \frac{p_i(b)}{k} \right)  \\
		&\geq \E \sum_i \left( u_i(b_{-i},b_i') + \frac{p_i(b)}{k} \right)  \\
		&\geq \E \sum_{S^* \in M^*} \sum_{i \in S^*} \left( u_i(b_{-i},b_i') + \frac{p_i(b)}{k} \right)  \\
		&\geq \E \sum_{S^* \in M^*} \sum_{i \in S^*} \left( u_i(b_{-i},b_i') + \sum_{j \in S^*} \frac{p_j(b)}{k^2} \right)  \\
		&\geq \E \sum_{S^* \in M^*} \sum_{i \in S^*} \frac{v_{iS^*}}{2k^2}
		\quad=    \frac{\Welf(\Opt)}{2k^2} .
	\end{align*}
\end{proof}
It remains to be seen if the factor can be improved to $\Omega(1/k)$.
We appear to lose one factor of $k$ because the greedy algorithm (e.g. the MM where all participants are truthful) is only a $1/k$ approximation to optimal, and then another factor from strategic behavior.

\subsection{Inspection} \label{sec:inspection}
We now extend our welfare result for graphs to a model with information acquisition costs.
We again do not require the graph to be bipartite.
The model is augmented as follows, following e.g. \citet{kleinberg2016descending}.
The type of $i$ consists of, for each feasible partner $j$, a cost of inspection $r_{ij} \geq 0$ and a distribution $D_{ij}$ over the nonnegative reals.
Our setting is Bayes-Nash, i.e. all types are drawn jointly from a common-knowledge prior.
When $i$ inspects an edge $\{i,j\}$, they incur a cost of $r_{ij}$ and observe a value $v_{ij} \sim D_{ij}$ independently of all other randomness in the game.
The cost $r_{ij}$ can model a financial investment, or the cost of time or effort required for $i$ to learn their value $v_{ij}$.

Let $I_{ij} \in \{0,1\}$ be the random variable indicator that $i$ inspects $j$ and let $A_{ij} \in \{0,1\}$ be the indicator that $i$ is matched to $j$.
We adopt the standard assumption (although recent algorithmic work of \citet{beyhaghi2019pandora} has weakened it) that $i$ must inspect $j$ prior to being matched to $j$; i.e. if the match occurs, $i$ must inspect and incur cost $r_{ij}$ if they haven't yet.
In other words, $A_{ij} \leq I_{ij}$ pointwise.
We also assume that, in the game, inspection is instantaneous with respect to the movement of the price $p(t)$.

An agent $i$'s utility is their value for their match (if any) minus the sum of all inspection costs and their net payment.
Formally, we have $u_i(b) = \sum_j \left(A_{ij} v_{ij} - I_{ij} r_{ij}\right) - \pi_i(b)$, where the sum is over feasible neighbors.
Social welfare is the sum of all agent utilities and revenue, i.e. $\Welf(b) = \sum_i \left( u_i(b) + \pi_i(b) \right)$.

\paragraph{Covered call values and exercising in the money.}
\citet{kleinberg2016descending} give technical tools, based on a solution of the Pandora's box problem~\citep{weitzman1979optimal}, utilizing finance-inspired definitions.\footnote{We refer the reader to \citet{kleinberg2016descending} for explanation of the terminology, but in brief, the idea is to imagine that when $i$ inspects $j$, the cost $r_{ij}$ is subsidized by an investor in return for a ``call option'', i.e. the right to the excess surplus $v_{ij} - \sigma_{ij}$ beyond a threshold $\sigma_{ij}$, if any. The agent's surplus for that match becomes $\kappa_{ij}$, and the investor breaks even if $(i,j)$ exercises in the money, otherwise loses money.}
\begin{definition} \label{def:strike-covered-call}
	Given a cost $r_{ij}$ and distribution $D_{ij}$, the \emph{strike price} is the unique value $\sigma_{ij}$ satisfying
	\[ \E_{v_{ij}\sim D_{ij}} \left(v_{ij} - \sigma_{ij}\right)^+ = r_{ij} , \]
	where $\left(\cdot\right)^+ = \max\{\cdot ~,~ 0\}$.
	The \emph{covered call value} is the random variable $\kappa_{ij} = \min\{\sigma_{ij}, v_{ij}\}$.
\end{definition}
We assume $\E_{v_{ij} \sim D_{ij}} v_{ij} \geq r_{ij}$ if $\{i,j\}$ is a feasible match.
As we have nonnegative values, this is equivalent to the condition $\sigma_{ij} \geq 0$.

A \emph{matching process} is any procedure that involves sequentially inspecting some of the potential matches and making matches, according to any algorithm or mechanism.
The following property, along with Lemma \ref{lemma:covered-call}, allows us to relate surplus in a matching process to that of a world with zero inspection costs and values $\kappa_{ij}$.
\begin{definition} \label{def:exercise-money}
	In any matching process, we say that the ordered pair $(i,j)$ \emph{exercises in the money} if, for all realizations of the process, if $v_{ij} > \sigma_{ij}$ and $I_{ij} = 1$ then $A_{ij} = 1$.
	We say that agent $i$ exercises in the money if $(i,j)$ exercises in the money for all feasible partners $j$.
\end{definition}

\begin{lemma}[Immediate extension of \cite{kleinberg2016descending}] \label{lemma:covered-call}
	For any feasible partners $\{i,j\}$, any fixed types of all agents, and any matching process,
	\[ \E\left[ A_{ij} v_{ij} - I_{ij} r_{ij} \right] \leq \E \left[ A_{ij} \kappa_{ij} \right], \]
	with equality if and only if $(i,j)$ exercises in the money.
\end{lemma}
\begin{proof}
	Using the definitions, independence of $v_{ij} \sim D_{ij}$ from the variable $I_{ij}$, and the assumption $A_{ij} \leq I_{ij}$,
	\begin{align*}
		\E \left[ A_{ij} v_{ij} - I_{ij} r_{ij} \right]
		&=    \E \left[ A_{ij} v_{ij} - I_{ij} \E_{v'_{ij} \sim D_{ij}} ( v_{ij}' - \sigma_{ij})^+ \right] \\
		&=    \E \left[ A_{ij} v_{ij} - I_{ij} ( v_{ij} - \sigma_{ij})^+ \right] \\
		&\leq \E \left[ A_{ij} v_{ij} - A_{ij} ( v_{ij} - \sigma_{ij})^+ \right] \\
		&=    \E \left[ A_{ij} \min\{\sigma_{ij}, v_{ij}\} \right].
	\end{align*}
	We observe that the inequality is strict if and only if there is positive probability of the event that $I_{ij} = 1$, $A_{ij} = 0$, and $v_{ij} > \sigma_{ij}$ all occur, i.e. $(i,j)$ fails to exercise in the money.
\end{proof}

\paragraph{The deviation strategy.}
Define the deviation strategy $b'_{i}$ for agent $i$ as follows: Initially bid $0$ on each neighbor $j$; when the clock reaches $\sigma_{ij}$, inspect neighbor $j$ and update bid to $\kappa_{ij} = \min\{\sigma_{ij}, v_{ij}\}$.
For a strategy profile $b$, define the random variable $\kappa_i(b)$ to be the covered call value of $i$ for its match in profile $b$, i.e. $\min\{\sigma_{ij}, v_{ij}\}$ when $i$ is matched to $j$.
Let $\bar{u}_i(b)$ denote $i$'s ``covered call utility'', i.e.
 \[ \bar{u}_i(b) = \kappa_i(b) - \pi_i(b). \]

\begin{lemma} \label{lemma:inspect-exercise}
	The deviation strategy $b_i'$ exercises in the money and ensures $\bar{u}_i(b_{-i},b_i') \geq 0$, for any $b_{-i}$.
\end{lemma}
\begin{proof}
	(Exercises in the money.)
	We consider the two possible scenarios where $i$ inspects a neighbor $j$, i.e. when $I_{ij} = 1$.
	If the inspection occurs because the mechanism has just matched $i$ to a previously-uninspected neighbor $j$, then $A_{ij} = 1$ and the requirement of exercising in the money is satisfied.
	Otherwise, the clock $p(t)$ has reached $\sigma_{ij}$, and after inspecting, $i$ updates that bid to $\kappa_{ij}$.
	If $v_{ij} \geq \sigma_{ij}$, then $b_{ij}' \geq p(t)$ and the match occurs immediately, so $A_{ij} = 1$.

	(Nonnegative covered call utility.)
	Recall that the rebate is nonnegative, so $i$'s net payment $\pi_i$ is always at most $i$'s bid.
	If $i$ is matched to some previously-uninspected $j$, then $b_{ij} = 0$, so $\pi_i(b_{-i},b_i') \leq 0$, and $\bar{u}_i(b_{-i},b_i') \geq 0$.
	If $i$ is matched to some $j$ that $i$ has already inspected, then $\pi_i(b_{-i},b_i') \leq b_{ij}' = \kappa_{ij}$, so $\bar{u}_i(b_{-i},b_i') \geq \kappa_{ij} - \kappa_{ij} \geq 0$.
	Finally, if $i$ is unmatched, then $\bar{u}_i(b_{-i},b_i') = 0$.
\end{proof}

\begin{lemma}[Covered call smoothness] \label{lemma:inspect-smooth}
  For any feasible neighbors $\{i,j\}$ and any strategy profile $b$, for all realizations of types and values:
  \[ \bar{u}_i(b_{-i},b_i') + \frac{p_i(b)}{4} + \frac{p_j(b)}{4} \geq \frac{\kappa_{ij}}{8} . \]
\end{lemma}
\begin{proof}
  Fix a realization of all types and values.
  The quantities $p_i(b)$, $p_j(b)$, and $\bar{u}_i(b_{-i},b_i')$ are all nonnegative.
  If $p_i(b) \geq \kappa_{ij}/2$ or $p_j(b) \geq \kappa_{ij}/2$, then we are already done.
  So suppose neither holds.
  Then in profile $b$, both $i$ and $j$ are not yet matched at price $p(t) = \kappa_{ij}$.
  Observe that in $(b_{-i},b_i')$, all other agents' bids and behavior are unchanged until $i$ is matched, because they continue playing their strategies in $b$ and no information available to them indicates the change in $i$'s strategy.

  Therefore, in profile $(b_{-i},b_i')$, $i$ matches at some $p(t) \geq \kappa_{ij}$, because if the price reaches $\kappa_{ij}$ without $i$ being matched, then $i$ will match to $j$.
  Let $j'$ be the partner $i$ is matched to.
  Then by construction, $i$ is bidding $b_{ij'}'(t) \leq \kappa_{ij'}$, so $\bar{u}_i(b_{-i},b_i') \geq \kappa_{ij'} - \kappa_{ij'} + \tfrac{p(t)}{4} \geq \tfrac{p(t)}{4} \geq \tfrac{\kappa_{ij}}{4}$.
\end{proof}

\begin{theorem} \label{theorem:inspect-welfare}
  In the inspection setting with nonnegative values, the $1/4$-rebate Marshallian Match guarantees a Bayes-Nash price of anarchy of at least $1/8$.
\end{theorem}
\begin{proof}
  Let the random variable $M^*$ be the max-weight matching where the weight on edge $\{i,j\}$ is $\kappa_{ij} + \kappa_{ji}$.
  Lemma \ref{lemma:covered-call} also implies that $\Welf(\Opt)$ is at most the total weight of $M^*$, as follows.
  Here $I_{ij},A_{ij}$ are the indicators under the $\Opt$ procedure, and notice $A_{ij} = A_{ji}$ in any matching.
  \begin{align*}
    \Welf(\Opt) &= \E \sum_{\{i,j\} \in E} \left[ \left( A_{ij} v_{ij} - I_{ij} r_{ij} \right) +  \left( A_{ji} v_{ji} - I_{ji} r_{ji} \right) \right] \\
    &\leq \E \sum_{\{i,j\}} \left(A_{ij} \kappa_{ij} + A_{ji} \kappa_{ji}\right)  & \text{Lemma \ref{lemma:covered-call}}\\
    &\leq \E \sum_{\{i,j\} \in M^*} \left( \kappa_{ij} + \kappa_{ji} \right)  &\text{the solution must form a matching.}
  \end{align*}
  (We note that we have no idea what $\Opt$ actually is in this setting, or if it can even be computed in polynomial time; nevertheless, this upper bound cannot be too loose, since the MM approximates it.)
  Lemma \ref{lemma:inspect-exercise} states that the deviation strategy $b_i'$ exercises in the money, so Lemma \ref{lemma:covered-call} implies $\E u_i(b_{-i},b_i') = \E \bar{u}_i(b_{-i}, b_i')$.
  Therefore,
  \begin{align*}
    \Welf(b) &= \E \sum_{i} \left[u_i(b) + \frac{p_i(b)}{2} \right] \\
    &\geq \E\sum_{i} \left[ u_i(b_{-i},b_i') + \frac{p_{i}(b)}{2} \right] \\
    &=    \E\sum_{i} \left[ \bar{u}_i(b_{-i},b_i') + \frac{p_{i}(b)}{2} \right] \\
    &\geq \E\sum_{\{i,j\}\in M^*} \left( \left[ \bar{u}_i(b_{-i},b_i') + \frac{p_i(b)}{4} + \frac{p_j(b)}{4} \right] + \left[ \bar{u}_j(b_{-j},b_j') +   \frac{p_{i}(b)}{4} + \frac{p_{j}(b)}{4} \right] \right) \\
    &\geq \E \sum_{\{i,j\}\in M^*} \left[ \frac{\kappa_{ij}}{8} + \frac{\kappa_{ji}}{8} \right]
    \quad \geq \frac{1}{8} \Welf(\Opt) .
  \end{align*}
\end{proof}

\section{General Bids} \label{sec:general-bids}

%

It is natural to consider negative values and bids in a matching market.
Negative values model \emph{costs} incurred for a match, such as in a job market.
When a worker is matched as an employee to a company, she experiences some cost for which she must be compensated.
The goal of the matching market is to find an efficient price for her labor.
However, negative costs complicate matters because they also introduce negative bids.
A participant can make it difficult for a match to occur by bidding an arbitrary negative amount.
The result is that equilibrium is no longer sufficient for good welfare, even in bipartite graphs.

In this section, we first formalize the failure of equilibrium.
We then turn to \emph{stability} as a possible saviour.
We observe that approximate \emph{ex post} stability indeed implies good welfare in \emph{any mechanism}, and that the MM satisfies approximate ex post stability when participants are truthful.

However, ex post stability is a very strong requirement.
We formulate an alternative, \emph{ex ante} stability, and show that the MM has approximately optimal welfare in any ex-ante stable \emph{Nash} equilibrium.
Finally, however, we observe that this result does not naturally extend to Bayes-Nash equilibrium and illustrate the apparent difficulty involving coordinated communication.

\subsection{Model and failure of equilibrium}
We make one restriction on the general model of Section \ref{sec:prelim}: we assume the graph is bipartite.
This is primarily for notational and narrative convenience.
To make our presentation more intuitive, we adopt terminology in which the two sides of the bipartite market are asymmetric: One side (e.g. employers) are \emph{bidders}, while the other side (e.g. workers) are \emph{askers}.
The bidders are indexed by $i$.
They have values $v_{ij}$ and make bids $b_{ij}$.

The askers are indexed by $j$.
We assume they have \emph{costs} $c_{ij}$ and make \emph{asks} $a_{ij}$.
The costs and asks are simply the negative of their values and bids under the previous section's model.
Now, for example the surplus of a match between $i$ and $j$ is $s_{ij} = v_{ij} - c_{ij}$.
We still use $\pi_j$ to denote the net payment made by an asker.
So the utility of an asker $j$ for being matched to $i$ is $u_j = -c_{ij} - \pi_i$, and so on.

We generally picture values, bids, costs, and asks all as positive numbers, so a bidder has a positive value for matching to an asker, who incurs a positive cost from the match.
However, our results are all fully general and would allow for any value, bid, cost, or ask to be either positive or negative.



\paragraph{Failure of equilibrium.}
When asks are allowed, equilibrium becomes insufficient to provide welfare guarantees.
Participants in the mechanism can place bids and asks such as to effectively refuse matches with one another, by asking above value or bidding below cost.
If two players both ``refuse'' matches with one another, neither can unilaterally fix the situation.
We prove this result for the MM with 1/4 rebate, but the same profile is also a zero welfare equilibrium with no rebate.

\begin{proposition} \label{prop:bad-welfare-bipartite}
	In the bipartite setting with general values and costs, there always exists a Nash equilibria of the Marshallian Match with zero welfare.
\end{proposition}
\begin{proof}
	Let $x = \max_{i,j} \max\{|v_{ij}|,|c_{ji}|\}$.
	Consider the strategy profile where every bidder $i$ bids $b_{ij} = -2x$ on all neighbors $j$, and every asker $j$ asks $a_{ji} = 2x$ on all neighbors $i$.
	This is an equilibrium in which no matches occur.
    For a unilateral deviation to cause a match to occur, e.g. a bidder would have to change a bid to at least $b_{ij} = 2x$, resulting in a net payment of at least $2x$ after the rebate, giving negative utility.
	The asker's case is analogous.
\end{proof}

In this example, any single player cannot deviate alone to improve her welfare.
However, any pair of bidders sharing an edge can coordinate a deviation together and guarantee themselves higher welfare.
This suggests that the bad equilibrium profile lacks \emph{stability}, a key concept in matching algorithm design.

\subsection{Ex Post Stability}
In classical ``stable matching'' problems~\cite{gale1962college}, the goal is that, once the mechanism produces a final matching, no two participants $i,j$ both prefer to leave their assigned partners and switch to matching each other instead.
We use \emph{ex post} to refer to the fact that this evaluation occurs after all randomness and the matching's outcome have been realized.
In our setting, if $i$ and $j$ chose to match each other, they would obtain a net utility of their surplus $s_{ij} = v_{ij} - c_{ij}$.
If this amount is larger than their total utility in the mechanism, then they could switch to each other and split the surplus so as to make them both better off.
On the other hand, making this switch presumably involves some amount of friction.
Therefore, we introduce an approximate version of stability, as a more lenient requirement of a mechanism.


\begin{definition}[Approximate ex post stability]
	A strategy profile $(b,a)$ in a mechanism is $k$-ex post stable if, for all realizations of costs and values and for all feasible pairs of bidder $i$ and asker $j$,
	\[u_i(b,a) + u_j(b,a)\geq \frac{1}{k}s_{ij} .\]
\end{definition}


We note that ex post stability is an extremely strong notion. For any matching mechanism, an ex post stable strategy profile produces an approximately optimal matching.

\begin{observation} \label{obs:ex-post-approx}
	For any matching mechanism, any $k$-ex post stable strategy is a $\tfrac{1}{k}$-approximation of the first-best welfare.
\end{observation}
\begin{proof}
	Recall that in the Bayes-Nash setting agents' types are drawn from from a joint distribution $\D$.
	Let $(b,a)$ be a $k$-ex post stable strategy profile for some matching mechanism, and $M$ the associated matching.
	The expected welfare over all realizations of types and strategies is at least the total utility of participants,
	\begin{align*}
		\Welf(b,a) \geq \E\left[\sum_{\{i,j\}\in M}u_i(b,a) + u_j(b,a)\right] = \E\left[\sum_{\{i,j\}\in M^*}u_i(b,a) + u_j(b,a)\right],
	\end{align*}
	where $M^*$ is the maximal matching by surplus.
	Ex post stability gives a lower-bound on utility for all pairs, regardless of realization, so
	\begin{align*}
		\Welf(b,a) \geq \E\left[\sum_{\{i,j\}\in M'}\frac{s_{ij}}{k}\right] =  \frac{1}{k}\Welf(\Opt)
	\end{align*}
\end{proof}
Observation \ref{obs:ex-post-approx} also arises directly from a primal-dual analysis of the linear program for bipartite matching, in which the dual variables are the utilities of the agents and $k$-approximate satisfaction of the dual constraints is precisely $k$-ex post stability.
We note that while there exist deferred-acceptance style ``stable'' matching mechanisms that technically involve money, such as matching with contracts~\cite{hatfield2005matching}, we have not ascertained if they can be made to satisfy approximate ex post stability in our sense.

\paragraph{Ex post stability of the Marshallian Match under truthfulness.}
Ideally, a matching mechanism would admit equilibria in ex post stable strategies, so that welfare would be high and participants would adhere to the outcomes of the mechanism.
But as a weaker requirement, we would like an indication of whether approximately ex post stable strategies might be reasonable in a mechanism.
We present a simple $4$-ex post stable strategy profile for the $1/4$-rebate MM, for all realizations of costs and values.

\begin{proposition} \label{prop:truthful-ex-post}
	The truthful strategy in which all participants bid or ask their true value or cost is $4$-ex post stable in the $1/4$-rebate Marshallian Match.
\end{proposition}
\begin{proof}
	When all participants truthfully report their values and costs, Marshallian Match produces the greedy maximum weighted matching $M$, where edge weights are supluses $s_{ij} = v_{ij} - c_{ij}$ (and negative edges are discarded).
	Let $(b^*,a^*)$ denote the strategy profile in which all agents are truthful.

	Consider any pair $\{i,j\}$.
	We wish to show that $u_i(b^*,a^*) + u_j(b^*,a^*) \geq \frac{s_{ij}}{4}$.
	If $\{i,j\}\in M$, then $i$ and $j$ each obtain utility exactly $s_{ij}/4$, as their value (cost) is canceled by their bid (ask) and they are left with the rebate.
	If $\{i,j\}\notin M$, then either $i$ or $j$ must have already matched before the descending price reached $s_{ij}$.
	Without loss of generality, say $i$ matched to $k$ before $s_{ij}$, with $s_{ik} > s_{ij}$.
	By truthfulness, $i$ obtains welfare $u_i(b^*,a^*) = s_{ik}/4 > s_{ij}/4$.
\end{proof}
Again, once one connects the truthful MM to greedy matching, Proposition \ref{prop:truthful-ex-post} follows from a basic primal-dual analysis combined with the rebate payment rule.

While truthful reporting from all participants produces approximately optimal social welfare (and the participants capture half of it), truthfulness is not in general even an approximate equilibrium.
We provide an example in the $1/4$-rebate MM setting in which a player can increase her welfare arbitrarily by overstating her true value for a match.
If players can make significant gains by deviating from the truthful strategy, it is likely that they will not adhere to an ex post stable profile.

\begin{example}[Figure \ref{fig:non-truthful}]
	Consider the bipartite graph with three participants: $A$ and $B$ place bids on matching to $C$.
	$C$ has cost $c_{CA} = c_{CB} = 0$ for both matches, $A$ has value $v_{AC} = k+1$ and $B$ has value $v_{BC} = k$ for matching with $C$.
	In the truthful profile, $A$ is matched to $C$, and $B$ goes unmatched with welfare $0$.
	If $B$ deviates to the non-truthful strategy of bidding $b'_{BC} = k+2$, $B$ would be matched with $C$ and would receive utility $u_B(b_{-B}^*,a^*,b') = k - (k+2) + (k+2)/4 = k/4-3/2$.
	Thus, picking $k$ appropriately, $B$ can benefit arbitrarily by deviating from the truthful strategy.
	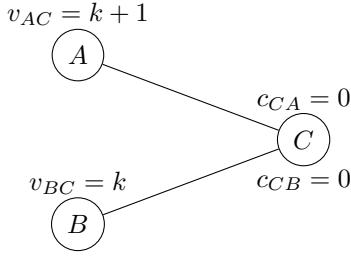
\begin{figure}
		\begin{center}
		\begin{tikzpicture}[
			every node/.style={circle,draw},scale=.75]

			\node (A) at (-2,1.5) {$A$};
			\node (B) at (-2,-1.5) {$B$};
			\node (C) at (2,0) {$C$};

			\draw (A) node[right,rectangle,above,sloped,draw=none,fill=none,outer ysep=8pt]{$v_{AC} = k+1$} -- (C) node[left,rectangle,above ,sloped,draw=none,fill=none,outer ysep=8pt]{$c_{CA} = 0$};
			\draw (B) node[right,rectangle,above,sloped,draw=none,fill=none,outer ysep=8pt]{$v_{BC} = k$} -- (C) node[left,rectangle,below ,sloped,draw=none,fill=none,outer ysep=8pt]{$c_{CB} = 0$};
		\end{tikzpicture}
		\end{center}
		\caption{In the $1/4$-rebate MM, $B$ is incentivized to deviate from truthfulness to the non-truthful bid $k+2$.}
		\label{fig:non-truthful}
	\end{figure}
\end{example}

\paragraph{Drawbacks of ex post stability.}
If one can prove that a mechanism is ex post stable in equilibrium, that is an ideal result as a welfare guarantee immediately follows.
However, without such a result, the value of ex post stability is questionable.
Consider the following somewhat subtle point.
Intuitively, it may seem reasonable to view stability as a sort of equilibrium refinement.
In particular, if strategy profile $b$ is not approximately stable, then (one would think) there are two participants $i$ and $j$ who would prefer to jointly switch their strategies so as to match to each other.
However, \emph{ex post} stability does not give this kind of guarantee.
It can only tell $i$ and $j$ whether they are satisfied after the mechanism happens.
To capture the above intuition, we turn to \emph{ex ante} stability.

\subsection{Ex Ante Stability}\label{subsec:ex_ante}
We now consider a model of stability that generalizes equilibrium by supposing no \emph{pair} of participants has an incentive to \emph{bilaterally} deviate.
Importantly, the incentive is relative to expected utility, so it involves an \emph{ex ante} calculation by the participants rather than ex post.

\begin{definition}[ex ante stability] \label{def:ex-ante}
	A strategy profile $(b,a)$ is $k$-ex ante stable if, for all feasible pairs of bidder $i$ and asker $j$, for all strategies $b_i'$ of $i$ and $a_j'$ of $j$,
	\[\E[u_i(b,a) + u_j(b,a)]\geq \frac{1}{k}\E[u_i(b_{-i},a_{-j},b_i',a_j') + u_j(b_{-i},a_{-j},b_i',a_j')] , \]
	with randomness taken over realizations of types and strategies.
\end{definition}
That is, a profile is $k$-ex ante stable if there exists no deviation for any pair of players by which they could expect to increase their collective welfare by a factor of more than $k$.
There are two primary differences between ex ante and ex post stability.
First, ex ante applies to preferences ``before the fact'', i.e. in expectation, while ex post applies to preferences ``after the fact''.
Second, ex post stability postulates the ability of two participants to completely bypass the rules of the mechanism and match to each other.
In ex ante stability, the participants are limited to deviating to strategies actually allowed by the mechanism.

We observe that an ex ante stable profile is by definition a Bayes-Nash equilibrium, as one can in particular consider profiles where only one of the two participants deviates.
We show that for deterministic values and costs, ex ante stability is in fact sufficient to guarantee an approximation of optimal welfare.

\paragraph{Smoothness and deviation strategies.}
We define a pairwise deviation of $\{i,j\}$ to truthfulness for bids and asks as follows:  $b_{i\ell}' = v_{i\ell}$ for all feasible neighbors $\ell$, and similarly $a_{j\ell}' = c_{j\ell}$ for all neighbors $\ell$.
For any pairwise deviation to truthfulness, we show that that pair collectively achieves a constant fraction of their welfare in the $1/4$-rebate MM.
\begin{lemma}\label{mechanism-splits-smoothness}
	For any pair of feasible neighbors $i$ and $j$, for any strategy profile $(b,a)$,
	\begin{equation*}
		u_i(b_{-i}a_{-j},b_i'a_j') + u_j(b_{-i}a_{-j},b_i'a_j') \geq \frac{s_{ij}}{4} .
	\end{equation*}
\end{lemma}

\begin{proof}
	If $i$ and $j$ are still unmatched in deviation when the price reaches $s_{ij}$, then they will be matched and each get utility $s_{ij}/4$, since their bid-ask spread is $s_{ij}$.

	If $i$ matches to some $\ell\neq j$ before $s_{ij}$, we claim $i$ still achieves high utility. Since $i$ is truthful, she obtains utility $u_i(b_{-i}a_{-j},b_i'a_j') = v_{i\ell} - b_{i\ell}' + (b_{i\ell}' - a_{\ell i})/4 = (b_{i\ell}' - a_{\ell i})/4 > s_{ij}/4$, since in deviation the bid-ask spread on the edge $(i,\ell)$ is greater than that on edge $(i,j)$.

	If $j$ matches to some $\ell\neq i$ before $s_{ij}$, then $j$ achieves high utility. Since $j$ is truthful, she obtains utility $u_j(b_{-i}a_{-j},b_i'a_j') = a_{\ell j}' - c_{\ell j}  + (b_{\ell j} - a_{j \ell}')/4 = (b_{\ell j} - a_{j \ell}')/4  > s_{ij}/4$, again since $(\ell,j)$ has a higher bid-ask spread in deviation than $(i,j)$.
\end{proof}

\begin{theorem} \label{thm:mechanism-splits-PoA}
	For deterministic values and costs, and strategy profile $(b,a)$ that is $k$-ex ante stable, the $1/4$-rebate Marshallian Match achieves a $\frac{1}{4k}$-approximation of the optimal expected welfare.
\end{theorem}
Surprisingly, in this result, the contribution of payments to overall welfare can be ignored.
We will actually obtain that total participant surplus is a constant fraction of the optimal welfare.

\begin{proof}[Proof of Theorem \ref{thm:mechanism-splits-PoA}]
	Let $M$ denote the matching produced by the strategy profile $(b,a)$, and $M^*$ denote the first-best welfare matching. The welfare of $M$ is
	\begin{align*}
			\textsc{Welf}(b,a)  &= \E\left[\sum_{\{i,j\}\in M}u_i(b,a) + u_j(b,a) + p_i(b,a)\right]\\
			 &\geq \E\left[\sum_{\{i,j\}\in M^*}u_i(b,a) + u_j(b,a)\right]\\
			&\geq \E\left[\frac{1}{k}\sum_{\{i,j\}\in M^*}u_i(b_{-i},a_{-i},b_i',a_j') + u_j(b_{-i},a_{-i},b_i',a_j')\right]
	\end{align*}
	since the values and costs of agents are deterministic and common knowledge, each agent can compute her partner in the first-best welfare matching, and can coordinate to deviate as a pair, guaranteeing a lower bound on welfare for both. Applying Lemma \ref{mechanism-splits-smoothness},
		\begin{align*}
			\textsc{Welf}(b,a) \geq \frac{1}{k}\E\left[\sum_{\{i,j\}\in M^*} \frac{s_{ij}}{4}\right] = \frac{1}{4k}\textsc{Welf}(\textsc{Opt})
	\end{align*}
\end{proof}

This result relies heavily on participants' ability to calculate the fixed matching with optimal social welfare, and coordinate deviations with their partner.
We exploit the deterministic costs and values considered to fix the optimal matching $M^*$ across realizations of potentially mixed strategies.

\subsection{Discussion: failure of the proof in the Bayes-Nash setting}
In this section, we discuss the main open problem: proving welfare guarantees in the Bayes-Nash setting under a reasonable stability definition.

\begin{openprob}
  Give a ``reasonable'' stability assumption and a mechanism for the matching model of \ref{sec:prelim}, with general values, such that, in the Bayes Nash setting, every ``stable'' strategy profile generates a constant factor of the optimal expected welfare.
\end{openprob}
In the Bayes-Nash setting, the distributions over values and costs are common knowledge, but their realizations are not.
As a result, the optimal matching $M^*$ is dependent on the realizations of players' types.
This causes our proof of welfare approximation under ex ante stability to fail in an interesting way.
In fact, even the definition of ex ante stability (Definition \ref{def:ex-ante}) has nuanced implications.

Consider a complete bipartite $n \times n$ graph with a ``star-crossed lovers'' distribution on types\footnote{One can create a version where type distributions are independent that makes roughly the same point.}: there is a uniformly random perfect matching $M^*$ in which the edges have positive surplus, while all edges not in $M^*$ have high costs and low values.
Even if we take a very bad mechanism, such as one that always assigns the same matching $M$ regardless of types, it can be ex ante stable: any particular pair are getting low utility, but if they are able to jointly deviate to matching with each other deterministically, they also get low utility in expectation over the type distribution.
They are unlikely to be fated for each other once the types are realized.

Similarly, recall that the proof of Theorem \ref{thm:mechanism-splits-PoA} used ex ante stability in the following step:
\begin{align*}
	\E\left[\sum_{(i,j)\in M^*}u_i(b,a) + u_j(b,a)\right]
	&\geq \E\left[\frac{1}{k}\sum_{(i,j)\in M^*}u_i(b_{-i},a_{-i},b_i',a_j') + u_j(b_{-i},a_{-i},b_i',a_j')\right] . 
\end{align*}
In the Nash setting, $M^*$ was fixed, and so were the deviation strategies $b_i',a_j'$ which depended on $M^*$.
So the expectation could move inside the sum, followed by an application of ex ante stability.
But in a Bayes-Nash setting, $M^*$ is a random variable depending on the realizations of the types.
We cannot move the expectation inside the sum.

A tempting fix is some sort of \emph{ex interim stability} assumption, where the deviation strategies of $i$ and $j$ can depend on the types of both players.
In the star-crossed lovers example, this is appealing: the pairs with high surplus know this fact from their types and can easily coordinate a deviation.
So it is reasonable to assume that strategy profiles played in the mechanism are not much worse than such coordinated deviations.

However, the amount of coordination required grows significantly if type distributions are more complicated.
For example, suppose each value and cost comes from an independent power-law distribution.
Finding blocking joint deviations seems to require significant knowledge by the agents, perhaps of global properties of the type space (such as who would be matched under the greedy matching or optimal matching).
To assume agents play strategies that eliminate such deviations appears to unjustly relieve the mechanism of responsibility to coordinate agents' information and decisionmaking.

\section{Conclusion and Future Work} \label{sec:conclusion}
Two-sided matching is difficult, even in such simplified abstract models as in this paper, for at least three reasons:
\begin{itemize}
  \item We would like the process to accommodate information acquisition in a way that is compatible with optimal search theory.
  \item The process is generally dynamic and sequential (for the previous reason), and strategic behavior in dynamic settings is complicated.
  \item The constraints are complex and interlocking, i.e. $i$ can match to $j$ if and only $j$ matches to $i$, yet their preferences over this event can be conflicting and contextualized by their other options.
\end{itemize}
The variants of the Marshallian Match studied in this paper address each challenge to some extent.
At least in the nonnegative values setting, the MM is compatible with optimal search because matches are coordinated to occur approximately in order from highest surplus to lowest.
This allows inspections to occur approximately in order of their ``index'' (``strike price'') from the Pandora's box problem.
In particular, it enables the useful technical property of ``exercising in the money'', i.e. a bidder who inspects at a late stage and discovers a very valuable match is able to unilaterally lock in that match.

Dynamic strategizing is addressed by strictly limiting information leakage, i.e. each participant can only see their own bid and learn when they match.
It is unclear whether this feature actually makes the mechanism better, but it does make it easier to analyze.
We believe that all of our results extend if participants are able to observe when any match occurs, but this would require careful formalization as the game becomes truly dynamic in that case.

The complex constraints are addressed to an extent by the coordination of matches in order of value.
A bidder $i$ in the $1/4$-rebate MM has deviation strategies available in which the timing of their match corresponds to their utility.
If another participant bids high enough to lock in a match with $i$ early, this is out of $i$'s direct control, but $i$ can bid so that they are assured of enough utility to make the match worthwhile.

\paragraph{Future work.}
There are a number of appealing variants on the model and directions for future investigation.
In the job market application, an interview is a \emph{simultaneous} inspection event between a worker and employer.
Can the MM's welfare guarantees extend to a model where inspection is simultaneous, even in the positive-values setting?
Other variations can include multiple stages of inspection (an extension in \citet{kleinberg2016descending}) or matching where inspection is optional (studied algorithmically by \citet{beyhaghi2019pandora}).

Another direction involves reasonable stability assumptions and their impact on monetary mechanism design for matching.
For example, a strategy profile seems somewhat unlikely if it allows for the following sorts of ``Stackelberg deviations'': participant $i$ announces a deviation strategy to all of their neighbors on the other side of the market, and those neighbors respond with their own deviations (perhaps best responses).

Finally, the MM itself admits a number of possible variations.
One additional benefit of the rebate payment rule is that it disincentivizes overbidding and under-asking, because participants do not need to strategically bid to capture value: the payment rule returns it to them as a rebate.
But is $1/4$ the optimal rebate level?
Our results seem to extend, and perhaps guarantee an even better welfare result, when the full payment is returned to bidders as a rebate.
However, such a mechanism seems suspicious, as bidders can manipulate their bids significantly without penalty.
Also interesting is the opposite extreme, a rebate of zero where all participants simply pay their bids.
We conjecture that the zero-rebate MM also achieves a constant PoA in the nonnegative bids setting.

A more sophisticated approach involves a two-part bid consisting of a reserve bid $b_{ij}$ and a ``surplus bid'' $\beta_{ij} \in (0,1)$ (respectively for askers, a reserve $a_{ji}$ and surplus $\alpha_{ji}$).
When a match occurs at $p(t) = b_{ij} - a_{ij}$, the bidder receives a rebate of $\beta_{ij} p(t)$ while the asker receives a rebate of $\alpha_{ji} p(t)$.
Perhaps a variant like this achieves a price of anarchy result, or a stability-based welfare guarantee, for the general setting.

\vfill
\break

\bibliographystyle{ACM-Reference-Format}
\bibliography{citations}


\end{document}